\documentclass[11pt,a4paper,reqno]{amsproc}

\makeatletter

\usepackage{datetime}
\usepackage[pdfborder={0 0 0}]{hyperref}
  \hypersetup{
    citebordercolor=.1 .6 1,
    colorlinks = false
}

\usepackage{upref,cases}
\hypersetup{citebordercolor=.1 .6 1}

\usepackage[T1]{fontenc}
\usepackage[utf8]{inputenc}

\usepackage[final]{graphicx}

\usepackage{geometry}
\IfFileExists{mmymtpro2.sty}{\usepackage[subscriptcorrection]{mymtpro2}
	\newcommand{\id}{\mathbf{1}}
	
  \def\cup{\cupprod}
  \def\cap{\capprod}
  \def\bigcup{\bigcupprod}
  
  \def\bigcupdisjoint{\mathop{\kern10pt\raisebox{4pt}{$\cdot$}\kern-12pt\bigcup}\limits}
}%
	{\usepackage{amssymb,bm,dsfont,fourier,times}
	 
	 \newcommand{\id}{\ensuremath{\mathds{1}}}
	\let\wtilde\widetilde
}

\let\epsilon\varepsilon
\newcommand{\eps}{\epsilon}

\renewcommand{\Im}{\ensuremath{\operatorname{Im}}}

\renewcommand{\emptyset}{\ensuremath{\varnothing}}

\DeclareMathOperator{\Tr}{Tr}
\DeclareMathOperator{\Int}{Int}

\DeclareMathOperator{\sign}{sgn}
\DeclareMathOperator*{\esssup}{ess\,sup}
\DeclareMathOperator*{\essinf}{ess\,inf}
\DeclareMathOperator{\supp}{supp}
\DeclareMathOperator{\dist}{dist}

\numberwithin{equation}{section}
\newtheoremstyle{ttheorem}%
       {1.8ex\@plus1ex}                
       {2.1ex\@plus1ex\@minus.5ex}      
       {\itshape}           
       {0pt}                   
       {\bfseries}          
       {.}                  
       {.5em}               
       {}                

\newtheoremstyle{ddefinition}%
       {1.8ex\@plus1ex}                
       {2.1ex\@plus1ex\@minus.5ex}      
       {}           
       {0pt}                   
       {\bfseries}           
       {.}                  
       {.5em}               
       {}                

\newtheoremstyle{rremark}%
       {1.8ex\@plus1ex}                
       {2.1ex\@plus1ex\@minus.5ex}      
       {\normalfont}        
       {0pt}                   
       {\bfseries}           
       {.}                  
       {.5em}               
       {}                   

\theoremstyle{ttheorem}
\newtheorem{theorem}{Theorem}[section]
\newtheorem{lemma}[theorem]{Lemma}

\newtheorem{corollary}[theorem]{Corollary}

\theoremstyle{ddefinition}

\theoremstyle{rremark}
\newtheorem{remark}[theorem]{Remark}
\newtheorem{myremarks}[theorem]{Remarks}
\newtheorem{myexamples}[theorem]{Examples}

\newenvironment{remarks}{\begin{myremarks}\begin{nummer}}%
    {\end{nummer}\end{myremarks}}
    {\end{nummer}\end{myexamples}}

\newcounter{numcount}
\newcommand{\labelnummer}{(\roman{numcount})}%

\makeatletter
\providecommand{\showkeyslabelformat}[1]{\relax}        
\let\mysaveformat\showkeyslabelformat                   %
\def\myformat#1{\raisebox{-1.5ex}{\mysaveformat{#1}}}   %

\newenvironment{nummer}%
  {\let\curlabelspeicher\@currentlabel%
    \begin{list}{\textup{\labelnummer}}%
      {\usecounter{numcount}\leftmargin0pt%
        \topsep0.5ex\partopsep2ex\parsep0pt\itemsep0ex\@plus1\p@%
        \labelwidth2.5em\itemindent3.5em\labelsep1em%
      }%
    \let\saveitem\item%
    \def\item{\saveitem%
      \def\@currentlabel{\curlabelspeicher\kern.1em\labelnummer}}%
    \let\savelabel\label%
    \def\label##1{{\ifnum\thenumcount=1\let\showkeyslabelformat\myformat\fi\savelabel{##1}}%
										{\def\@currentlabel{\labelnummer}%
									 	\let\showkeyslabelformat\@gobble
									 	\savelabel{##1item}%
										}%
	   							}%
  }{\end{list}}%

  {\let\curlabelspeicher\@currentlabel%
    \begin{list}{\textup{\labelnummer}}%
      {\usecounter{numcount}\leftmargin0pt%
        \topsep0.5ex\partopsep2ex\parsep0pt\itemsep0ex\@plus1\p@%
        \labelwidth2.5em\itemindent0em\labelsep1em%
        \leftmargin2.5em}%
    \let\saveitem\item%
    \def\item{\saveitem%
      \def\@currentlabel{\curlabelspeicher\kern.1em\labelnummer}}%
    \let\savelabel\label%
    \def\label##1{{\ifnum\thenumcount=1\let\showkeyslabelformat\myformat\fi\savelabel{##1}}%
										{\def\@currentlabel{\labelnummer}%
									 	\let\showkeyslabelformat\@gobble
									 	\savelabel{##1item}%
										}%
    							}%
  }{\end{list}}%

\usepackage{enumitem}

\let\OldItem\item
\newcommand{\MyItem}[2][]{}%
\newenvironment{MyDescription}[1][]{%
    \renewcommand{\item}[2][]{%
        \begin{enumerate}[#1,label={##1},ref={##1}]%
            \OldItem {##2}%
        \end{enumerate}%
    }%
}{%
}%

\def\section{\@startsection{section}{1}%
  \z@{1.3\linespacing\@plus\linespacing}{.5\linespacing}%
  {\normalfont\bfseries\centering}}

\def\subsection{\@startsection{subsection}{2}%
  \z@{.8\linespacing\@plus.5\linespacing}{-1em}%
  {\normalfont\bfseries}}

\def\nlsubsection{\@startsection{subsection}{2}%
  \z@{.8\linespacing\@plus.5\linespacing}{.1ex}%
  {\normalfont\bfseries}}

\let\@afterindenttrue\@afterindentfalse%

\renewenvironment{proof}[1][\proofname]{\par \normalfont
  \topsep6\p@\@plus6\p@ \trivlist 
  \item[\hskip\labelsep\scshape
    #1\@addpunct{.}]\ignorespaces
}{%
  \qed\endtrivlist
}

\def\ps@firstpage{\ps@plain
  \def\@oddfoot{\normalfont\scriptsize \hfil\thepage\hfil
     \global\topskip\normaltopskip}%
  \let\@evenfoot\@oddfoot
  \def\@oddhead{
    \begin{minipage}{\textwidth}
      \normalfont\scriptsize
      \emph{\insertfirsthead}
    \end{minipage}}
  \let\@evenhead\@oddhead 
}

\def\insertfirsthead{}

\def\@cite#1#2{{%
 \m@th\upshape\mdseries[{#1}{\if@tempswa, #2\fi}]}}
\renewcommand{\H}{\mathcal{H}}
\newcommand{\C}{\mathbb{C}}
\newcommand{\N}{\mathbb{N}}
\newcommand{\PP}{\mathbb{P}}
\newcommand{\R}{\mathbb{R}}
\newcommand{\Z}{\mathbb{Z}}
\renewcommand{\le}{\leqslant}
\renewcommand{\ge}{\geqslant}
\renewcommand{\leq}{\leqslant}
\renewcommand{\geq}{\geqslant}
\providecommand{\wtilde}[1]{\widetilde{#1}}

\let\<\langle
\let\>\rangle

\newcommand{\1}{1}

\newcommand{\upd}{\mathrm{d}}
\renewcommand{\d}{\upd}   

\newcommand{\hairspace}{\kern .04167em}

\renewcommand{\S}{\mathcal{S}}

\newcommand{\Zd}{\mathbb{Z}^d}
\newcommand{\Rn}{\mathbb{R}^d}

\DeclareMathOperator{\TextIm}{Im}

\renewcommand{\Im}{\TextIm}

\newcommand{\beq}{\begin{equation}}
\newcommand{\eeq}{\end{equation}}
\newcommand{\be}{\begin}
\newcommand{\e}{\end}

\newcommand{\E}{\mathbb E}

\def\clap#1{\hbox to 0pt{\hss#1\hss}}

\makeatother
\begin{document}


\title[Lower bound on the density of states]{A lower Wegner estimate and bounds on the spectral shift function for continuum random Schr\"odinger operators}

\author[M.\ Gebert]{Martin Gebert}
\address[M.\ Gebert]{School of Mathematical Sciences, Queen Mary University of London, Mile End
Road, London E1 4NS, UK}
\email{m.gebert@qmul.ac.uk}
\thanks{Supported in part by the European Research Council starting grant SPECTRUM (639305)}

\begin{abstract}
We prove a strictly positive, locally uniform lower bound on the density of states (DOS) of continuum random Schr\"odinger operators on the \textit{entire spectrum}, i.e. we show that the DOS does not have a zero within the spectrum. 
This follows from a \textit{lower} Wegner estimate for finite-volume continuum random Schr\"odinger operators. We assume throughout iid random variables and the single-site distribution having a Lebesgue density bounded from below on its support.
 The main mathematical novelty in this paper are pointwise-in-energy bounds on the expectation of the spectral shift function at all energies for these
operators where we mainly focus on perturbations corresponding to a change from Dirichlet to
Neumann boundary conditions along the boundary of a cube. We show that the bound scales with the area of the hypersurface where the boundary conditions are changed. We also prove bounds on the averaged spectral shift function for perturbations by bounded and compactly supported multiplication operators. 
\end{abstract}

\maketitle

\section{Introduction and results}

Wegner estimates are among the most frequently used tools in the theory of random Schr\"odinger operators. For example, they imply that the \textit{integrated density of states} (IDOS) of random Schr\"odinger operators with a suitably regular single-site probability distribution is Lipschitz continuous and therefore has a bounded density, which is called the \textit{density of states} (DOS), see \cite{MR2307751, MR2378428} for reviews. In his celebrated work \cite{Wegner} Wegner proved an upper and a positive lower bound on the DOS for the Anderson model on $\Z^d$. His initial motivation for proving such bounds was to clarify that a possible phase transition in the spectrum of such operators is not associated with a singularity or zero of the DOS.
In the following, the upper bound attracted a lot of attention in the mathematical physics community as it turned out to be useful for proving localisation for random Schr\"odinger operators via the multiscale analysis method, see \cite{AizWarBook} for the history of localisation. 

In this paper we consider (alloy-type) continuum random Schr\"odinger operators, see Section~\ref{themodel} for precise definitions. For these operators an optimal upper Wegner estimate is obtained in \cite{MR2362242}. On the other hand, a lower Wegner estimate and a strictly positive lower bound on the DOS for these models were only obtained recently in the region of \textit{complete localisation} in \cite{doi:10.1093/imrn/rnx092}.   However, for the Anderson model on $\Z^d$ a positive lower bound on the DOS holds on the \textit{entire spectrum}. This was proved in \cite{Jeske92,artRSO2008HiMu}  using the ideas of \cite{Wegner}.
In this paper we show that localisation is not necessary for a lower Wegner estimate to hold for continuum random Schr\"odinger operators as well. More precisely, we show a lower Wegner estimate which implies strict positivity of the DOS on the \textit{entire} spectrum except possibly very close to band edges, see Theorem \ref{th:LowerWegner} and Corollary \ref{cor:LowerBoundDOS}, where the constant in the lower Wegner estimate is locally uniform in energy. Our strategy of proof is the same as in \cite{doi:10.1093/imrn/rnx092} which relies on a Dirichlet-Neumann bracketing argument and controlling the resulting error for the spectral shift function. At this point the restriction to the region of \textit{complete localisation} was necessary in \cite{doi:10.1093/imrn/rnx092}. 
 
The main new ingredient and mathematical novelty presented in this paper are quantitative bounds on the averaged finite-volume spectral shift function for continuum random Schr\"odinger operators. In Theorem
 \ref{lem:SSFEstimate} we obtain new pointwise-in-energy bounds
on the averaged finite-volume spectral shift function corresponding to a change from Dirichlet to
Neumann boundary conditions along the boundary of a finite volume. Our bound is optimal in the sense that it is proportional to the area of the hypersurface where the boundary conditions are changed and is uniform in the finite-volume restriction. 
Compared to earlier results obtained in \cite{doi:10.1093/imrn/rnx092}, which relied on exponential decay of fractional moments of resolvents,  
the proof here relies 
only on a priori fractional moment bounds for continuum random Schr\"odinger operators established in \cite{artRSO2006AizEtAl2}.
We also show pointwise-in-energy bounds for the finite-volume spectral shift function for perturbations by bounded and compactly supported potentials, see Theorem \ref{thmgenW}. Bounds on the spectral shift function in the context of multi-dimensional Schr\"odinger operators have a long history and we refer to Remark \ref{ssf-problems} for related results.

\subsection{The model}\label{themodel}
We consider random Schr\"odinger operators with an alloy-type random potential of the form
\begin{equation}
\label{eq:TheOperator}
\omega \mapsto  H_0+V_\omega := H_0 + \sum_{k\in\Zd} \omega_k u_k
\end{equation}
acting on a dense domain in the Hilbert space $L^{2}(\Rn)$ for $d\in\N$. $H_0$ is a non-random self-adjoint operator and $\omega \mapsto V_{\omega}$ is a random potential subject to: 
\be{MyDescription}
\item[(K)]{ The unperturbed operator is given by $H_0 := -\Delta+V_0$ with $-\Delta$ being the
	non-negative Laplacian on $d$-dimensional Euclidean space $\Rn$ and $V_0 \in L^{\infty}(\R^{d})$  is a deterministic, $\Zd$-periodic and bounded
	background potential. We also assume $H_0\ge 0$.   \label{assK}}
\item[(V1)]{ The random coupling constants $\omega:=(\omega_k)_{k\in\Zd} \in \R^{\Z^{d}}$ are
	identically and independently distributed according to the Borel probability measure $\PP := \bigotimes_{\Z^{d}}P_{0}$
	on $\R^{\Z^{d}}$.	 The single-site distribution $P_0$ is absolutely continuous with respect to Lebesgue measure on $\R$. The corresponding Lebesgue density
 $\rho$ is bounded with support $\supp(\rho) \subseteq [0,1]$. We denote the corresponding expectation by $\E[\cdot]$.
 	\label{assV1}  }
\item[(V2)]{ The single-site bump functions $u_k(\,\cdot\,):=u(\,\cdot\,-k)$, $k\in\Zd$, are translates
	of a non-negative bounded function $0 \le u\in L^\infty_{c}(\R^d)$ supported in a ball of radius
	$R_u>0$. Moreover, we assume a covering condition, i.e. there exist  $C_{u,-}, C_{u,+}>0$ such that
	\begin{equation}
		\label{uk-bounds}
		0< C_{u,-} \leq \sum_{k\in\Zd} u_k \leq C_{u,+}<\infty.
	\end{equation}
	\label{assV2}
	}

\end{MyDescription}

We note that the condition $\supp(\rho) \subseteq [0,1]$ in (V1) is not stronger than the seemingly weaker property $\supp(\rho)$ is compact. In fact, the former can be obtained from the latter
via the inclusion of an additional periodic potential, a change of variables of the random couplings $(\omega_k)_{k\in\Z^{d}}$ and by rescaling the single-site potential $u$.	The random potential
	$V$ need not even be of the precise form \eqref{eq:TheOperator}, as $\Z^{d}$-translation
	invariance is not necessary for most of the arguments which do not involve the existence of the IDS.

The above model is $\Zd$-ergodic with respect to lattice translations. It follows that there exists a closed set $\Sigma\subset\R$, the non-random spectrum of $H$, such that $\Sigma=\sigma(H)$ holds $\mathbb{P}$-almost surely \cite{pastfig1992random}. We drop the subscript $\omega$ from $H$ and other quantities when we think of these quantities as random variables.
The covering conditions \eqref{uk-bounds} imply the spectral inclusion
\begin{equation}
 	\Sigma_{0} + [0, C_{u,-}] \subseteq \Sigma \subseteq \Sigma_{0} + [0, C_{u,+}],
\end{equation}
where $\Sigma_{0} := \sigma(H_{0})$, see \cite{pastfig1992random}.

Given an open set $G\subset\Rn$, we write $H_G$ for the Dirichlet restriction of $H$ to $G$.
We define the random finite-volume eigenvalue counting function
\beq
 \R \ni E \mapsto N_L(E):= \Tr \left(\id_{(-\infty,E]}(H_L)\right)
\eeq
for $L>0$. Here, $\id_{B}$ is the indicator function of a Borel set $B\subset \R$, $H_L:=H_{\Lambda_L}$ and $\Lambda_L:=(-L/2,L/2)^{d}$ is the open cube about the origin of side-length $L$. We also define $\Lambda_L(x_0):= x_0 + \Lambda_L$.  A Wegner estimate holds under our assumptions: given a bounded interval $I\subset\R$ and $E_1,E_2\in I$ with $E_1<E_2$, the estimate
\beq
\label{eq:WegnerEst11}
\mathbb{E}\big[N_L(E_2) - N_L(E_1)\big] \le C_{W,+}(I) |\Lambda_L| (E_2-E_1)
\eeq
holds for all $L>0$, where $C_{W,+}(I)$ is a constant which is polynomially bounded in $\sup I$, and
$|B|$ is the Lebesgue measure of a Borel set $B\subseteq \Rn$.
We refer to \cite{MR3046996,MR3106507,MR3749376} for recent trends concerning Wegner estimates.
Ergodicity implies, almost surely, that the limit
\beq
\label{eq:DOSFinVol}
N(E):= \lim_{L\to\infty} \frac 1 {|\Lambda_L|} N_L(E)
\eeq
exists for all $E\in\R$, see \cite{pastfig1992random}. The non-random limit function $N$ is called the integrated density of states (IDOS) of $H$, see e.g. the reviews \cite{MR2307751, MR2378428}. The Wegner estimate \eqref{eq:WegnerEst11} implies Lipschitz continuity of the IDOS $N$, which implies absolute continuity of $N$ with a bounded Lebesgue density $n$. The latter is the density of states (DOS) of $H$. The Wegner bound for the DOS implies
\begin{equation}
\label{eq:WegnerEst22}
\esssup_{E\in I} n(E) \leq C_{W,+}(I).
\end{equation}

\subsection{Lower Wegner estimate and strict positivity of the DOS}

One goal of this paper is to derive  lower bounds for the IDOS and DOS of alloy-type random Schr\"odinger operators that complement \eqref{eq:WegnerEst11} and \eqref{eq:WegnerEst22}.
To state the theorem, we need an additional assumption on the single-site probability density.
\be{MyDescription}
\item[(V1')]{ The single-site probability density is bounded away from zero on its support
	\begin{equation}
	  \rho_-:=\essinf_{\nu\in[0,1]}\rho(\nu) > 0.
	\end{equation}}
\end{MyDescription}

In the following we show a lower bound on the expectation of the eigenvalue counting function which we refer to as a \emph{lower} Wegner estimate. We use the notation $\Int(A)$ for the interior of a set $A\subset \R$.

\begin{theorem}
\label{th:LowerWegner}
Assume {\upshape(K)}, {\upshape(V1)}, {\upshape(V1')} and {\upshape(V2)}. Consider a compact energy interval $I\subset \Int(\Sigma_0+ [0,C_{u,-}])$. Then there exists a constant $C_{W,-}(I)>0$ and an initial length scale $L_0>0 $ such that
\beq
\label{eq:LowerWegnerStat}
\mathbb{E}\left[ \Tr\big(\id_{[E_1,E_2]}(H_L)\big) \right] \geq C_{W,-}(I) (E_2-E_1) |\Lambda_L|
\eeq
holds for all $E_1,E_2\in I$ with $E_1<E_2$ and all $L> L_0$.
\end{theorem}

\begin{corollary}
	\label{cor:LowerBoundDOS}
Assume {\upshape(K)}, {\upshape(V1)}, {\upshape(V1')} and {\upshape(V2)}. Consider a compact energy interval $I\subset  \Int(\Sigma_0+ [0,C_{u,-}])$.	 Then there exists a constant
	$C_{W,-}(I)>0$ such that
	\beq
		\label{eq:LowerBoundDOS}
		\essinf_{E\in I}\, n(E) \geq C_{W,-}(I).
	\eeq
\end{corollary}

\begin{remarks}
\item
If $V_0=0$, $\Int(\Sigma_0+ [0,C_{u,-}]) = [0,\infty)$ and the latter results hold on the \textit{entire} almost sure spectrum. The same is true for $C_{u,-} =C_{u,+}$. The restriction to $\Int\big(\Sigma_0+ [0,C_{u,-}]\big)$ and the assumption {\upshape(V1')} originates in the strategy of proof in \cite{doi:10.1093/imrn/rnx092} which is also used here. It consists of giving a lower bound on the DOS of $H$ in terms of the IDOS of $H_0$ at certain energies. However, the argument fails for energies which are far away from the spectrum of the operator $H_0$.
\item 
Partially, the interest in lower bounds on the DOS stems from the occurrence of the DOS as the intensity of the Poisson point process describing level statistics of eigenvalues in the localised regime. 
Poisson statistics are well understood for the Anderson model on $\Z^d$ \cite{MR1385082}.
It was also believed to hold for continuum models on $\R^d$ as well \cite{MR2663411}. Very recently, using new ideas, this was established at the very bottom of the spectrum of continuum random Schr\"odinger operators in \cite{DietleinElgart}. 
\end{remarks}

The proof of a lower bound on the DOS for alloy-type continuum random Schr\"odinger operators
 requires detailed control of the spectral shift function for a perturbation by a boundary condition on a hypersurface. This control is the main achievement of this paper. 

\subsection{Bounds on the spectral shift function}
 Let $A,B,C$ be self-adjoint operators which are bounded from below and admit purely discrete spectrum. Then the spectral shift function (SSF) for $A$ and $B$ is defined by
\begin{equation}
\label{eq:DefSSF}
\xi\big(E,A,B\big) := \Tr \big(\id_{(-\infty,E]}(A)-\id_{(-\infty,E]}(B) \big).
\end{equation}
This definition makes sense and coincides with the abstract definition of the SSF \cite{yafaev1992mathscattering}.
Note that the spectral shift function is linear with respect to the perturbation, i.e.
\beq\label{linear}
\xi\big(E,A,B\big) = \xi\big(E,A,C\big) + \xi\big(E,C,B\big).
\eeq
For $L,l\in\mathbb{R}_{>0}$ and $x_{0} \in \Lambda_{L}$ such that $\overline{\Lambda_{l}(x_{0})}\subset\Lambda_L$ we write $H^{D/N}_{L,l}$, for the self-adjoint restriction of the operator $H$ to
$\Lambda_{L}\setminus\overline{\Lambda_l(x_{0})}$ with Dirichlet boundary conditions along the outer boundary $\partial\Lambda_{L}$ and Dirichlet boundary conditions, respectively Neumann boundary conditions along the inner boundary $\partial \Lambda_{l}(x_{0})$. Moreover, we denote by $H_L^{D/N}$ the restriction of $H$ to $\Lambda_L$ with Dirichlet, respectively Neumann boundary conditions.
Then we have the following bounds on the corresponding spectral shift functions.

\begin{theorem}
	\label{lem:SSFEstimate}
Assume {\upshape(K)}, {\upshape(V1)} and {\upshape(V2)}. Given a compact energy interval $I\subset \R$ then
\begin{itemize}
	\item[(i)]\label{i} there exists $C_1>0$ such that
	\begin{equation}
	\label{eq:SSFEstStat}
	\sup_{E\in I}\mathbb{E}\big[\big|\xi(E,H^{N}_{L,l},H^{D}_{L,l})\big|\big] \leq C_1 l^{d-1}
	\end{equation}
	holds for all $L >1$ and $l\ge 1$ such that $x_{0} \in \Lambda_{L}$, provided $\Lambda_l(x_{0})\subset \Lambda_L$ with $						\dist\big(\partial\Lambda_l(x_{0}), \partial \Lambda_L\big)\geq 1$.
	\item[(ii)] there exists $C_2>0$ such that
		\begin{equation}
		\sup_{E\in I}\mathbb{E}\big[\big|\xi(E,H^{N}_{L},H^{D}_{L})\big|\big] \leq C_2 L^{d-1}.
	\end{equation}
\end{itemize}	
\end{theorem}

Using the operator inequalities $H^N_{L,l} \oplus H^N_{\Lambda_{l}(x_{0})} \le H_L\le H^D_{L,l} \oplus H^D_{\Lambda_{l}(x_{0})}$ and $\xi(E,A,B)\le \xi(E,A,C)$ for $B\le C$, Theorem \ref{lem:SSFEstimate} implies

\begin{corollary}
\label{cor:SSF}
Assume {\upshape(K)}, {\upshape(V1)} and {\upshape(V2)}. Given a compact energy interval $I\subset \R$, then there exists a constant $C>0$ such that
\begin{equation}
	\label{eq:SSFEstStat2}
	\sup_{E\in I}\mathbb{E}\big[\big|\xi\big(E,H_{L},H^{\star}_{L,l} \oplus H^\star_{\Lambda_{l}(x_{0})}\big)\big|\big] \leq C l^{d-1}
	\end{equation}
	holds for all $L >1$ and $l\ge 1$ such that $x_{0} \in \Lambda_{L}$, provided $\Lambda_l(x_{0})\subset \Lambda_L$ with $\dist\big(\partial \Lambda_l(x_{0}), \partial \Lambda_L\big)\geq 1$. 
	Here, $\star$ denotes boundary conditions on the inner boundary $\partial \Lambda_l(x_{0})$ subject to $H^{N}_{L,l} \oplus H^N_{\Lambda_{l}(x_{0})}\le H^{\star}_{L,l} \oplus H^\star_{\Lambda_{l}(x_{0})}\le H^{D}_{L,l} \oplus H^D_{\Lambda_{l}(x_{0})}$ .
\e{corollary}

\begin{remark}
We restricted ourselves to $L>1$,  $l\ge 1$ and boxes with $\dist\big(\partial \Lambda_l(x_{0}), \partial \Lambda_L\big)\geq 1$ in the above theorems since we do not want to complicate things with notational subtleties arising from small length scales $l,L$ in the proof.  
\end{remark}

\begin{theorem}\label{thmgenW}
Assume {\upshape(K)}, {\upshape(V1)} and {\upshape(V2)}. Let $W\in L^\infty_{c}(\R^d)$ .
Given a compact energy interval $I\subset \R$, there exists a constant $C>0$ such that 
\beq
\sup_{E\in I} \E\big[\big| \xi(E,H_L,H_L+W)\big|\big]<C
\eeq
for all $L>1$. 
\end{theorem}

The proof of Theorem \ref{thmgenW} follows along the very same lines as Theorem \ref{lem:SSFEstimate}. We indicate it in Section \ref{proofremaining}. 

\be{remarks}
\label{ssf-problems}
\item
	Theorem \ref{lem:SSFEstimate} (ii) implies that the IDOS is independent of any boundary conditions lying in between
	Dirichlet and Neumann boundary conditions. 
	This is known for deterministic Schr\"odinger operators with magnetic fields
	\cite{Nakamura:2001jh,Doi:2001hy,Mine:2002tl}  provided the IDOS exists. The deterministic  proofs there show that the error is of order $o(|\Lambda_L|)$.
\item
	For continuum Schr\"odinger operators, there is a subtle problem to obtain
	bounds on the SSF,
	which hold pointwise in energy and are uniform in the volume. It has been noted in
	\cite{artSSF1987Kir}  that for  $d\geq 2$ and  $E>0$
	\beq
		\label{kirsch-div}
		\sup_{L>0} \xi\big(E, -\Delta^D_L,-\Delta^D_{L}+W\big) =\infty
	\eeq
	for $W\in L^\infty_c(\R^d)$ and $W\ge 0$. 
	The divergence in \eqref{kirsch-div} is due to increasing degeneracies of the eigenvalues of the Laplacian for
	larger volumes at fixed energy. Theorem \ref{thmgenW} shows that such subtleties are ruled out by the randomness. 
	and one has pointwise-in-energy bounds on the SSF uniform in the volume and locally uniform in energy. 
\item
	There are several results for $L^{p}$-boundedness of the SSF \cite{MR1824200,MR1945282,MR2200269}.
	Results analogous to Theorem \ref{lem:SSFEstimate} and \ref{thmgenW} are known in the region of complete localisation, see \cite{doi:10.1093/imrn/rnx092} and \cite[Thm.~3.1, Cor.~3.2]{DiGeMu16b}. 
	For a special class of perturbations pointwise-in-energy bounds on the averaged SSF are proved in \cite{MR2352262}. 
\item	
 A bound on a one-parameter averaged fractional spectral shift function is proved in  \cite[App. D]{artRSO2006AizEtAl2} for a perturbation by a potential. Their ideas, using the Birman-Schwinger principle and a weak-$L^1$ bound on the resolvent, are quite similar to the ideas used in our proof. 
\end{remarks}

\section{Proof of Theorem \ref{lem:SSFEstimate}}
\label{sec:SSFEst}
Let  $p>0$. Then we denote by $\S^p$ the corresponding Schatten-$p$-class with norm $\|\cdot\|_p$. For $0<p<1$ and a compact operator $A$  with singular values $(a_n)_{n\in\N}$ we note that $\|A\|_p:=\big(\sum_{n\in\N} |a_n|^p\big)^{1/p}$  is not a norm but a quasi norm. However, the
 ``triangle-like''
inequality
\beq
\label{ptriangle}
\|A+B\|_p^p\leq \|A\|_p^p+\|B\|_p^p
\eeq
still holds for compact operators $A,B$ and $p\in (0,1]$, see \cite[Thm. 2.8]{McCarthy}. 
We also note that the (generalised) H\"older inequality for Schatten classes remains true for H\"older exponents $p_1,...,p_n>0$ subject to $p_1^{-1}+...+p^{-1}_{n}=p^{-1}$.
We write $\chi_x:=\id_{\Lambda_1(x)}$ for $x\in\R^d$, where $\id_A$ stands for the indicator function on the Borel set $A\subset\R^d$.

We only prove Theorem \ref{lem:SSFEstimate} (i). The second part, Theorem \ref{lem:SSFEstimate} (ii), follows along the very same lines. 
We restrict ourselves to $l\in\N$. Let $\Lambda_l(x_0)\subset \Lambda_L$ for some $x_0\in \Lambda_L$. 
We denote by $\partial \Lambda_l$ the inner boundary of $\Lambda_L\backslash \Lambda_l(x_0)$ and we set $\Lambda_{L,l}:=\Lambda_L\backslash\overline{ \Lambda_l(x_0)}$. 

\subsection{Splitting the perturbation into $O(l^{d-1})$ smaller perturbations}
We assume $l\in\N$ and split $\partial \Lambda_l$ in a disjoint union of hypercubes of side length $1$. More precisely, there exists a constant $C_l\in \N$ such that
\beq\label{decomp}
\partial \Lambda_l = \bigcup_{n=1,...,C_l} \overline{\Gamma_n},
\eeq
for a family $\big(\Gamma_n\big)_{n=1}^{C_l}$ of  open hypercubes of side length $1$ such that $\Gamma_n\cap\Gamma_m=\emptyset$. The constant $C_l$ satisfies $C_l= O(l^{d-1})$, as $l\to\infty$. 

Next we define a family of operators lying between $H^{N}_{L,l}$ and $H^{D}_{L,l}$. For $M\in\{1,..., C_l\}$ we define 
$H_{L,l,M}$ as the operator which admits Dirichlet boundary conditions on the set $\Gamma_M:=\displaystyle\bigcup_{n=1}^M \Gamma_n$ and Neumann boundary conditions on the complement $\displaystyle\partial \Lambda_l\backslash \bigcup_{n=1}^M \overline{\Gamma_n}$. More precisely, the corresponding Laplace operator is defined by the quadratic form 
$\int_{\Lambda_{L,l}} \d x \overline{\nabla u}\,\nabla v $ with domain 
\beq
H_{0,\Gamma_M}(\Lambda_{L,l}):=\big\{u\in H^1(\Lambda_{L,l}): u\big|_{\Gamma_M}=0\big\}.
\eeq
Here $H^1(\Lambda_{L,l})$ is the first order Sobolev space on $\Lambda_{L,l}$. 
 For more information on Laplace operators with mixed boundary conditions we refer to \cite{MR3631314} and references therein. 
We use the notation 
$H_{L,l,0} := H^{N}_{L,l}$.

Now we exploit the linearity of the spectral shift function with respect to the perturbation, see \eqref{linear}, and obtain for $E\in \R$ that
\beq
\xi(E,H^{N}_{L,l},H^{D}_{L,l}) = \sum_{M=1}^{C_l} \xi(E,H_{L,l,M},H_{L,l,M-1}) .
\eeq
This implies
\beq
\big|\xi(E,H^{N}_{L,l},H^{D}_{L,l}) \big|\leq  \sum_{M=1}^{C_l}\big| \xi(E,H_{L,l,M},H_{L,l,M-1})\big|.
\eeq
Since $C_l= O(l^{d-1})$, as $l\to\infty$, the assertion follows once we have proved
\beq\label{pf:eq8}
\sup_{M\in\{1,..., C_l\}}\E\big[\big| \xi(E,H_{L,l,M},H_{L,l,M-1})\big|\big] < C
\eeq
for a constant $C>0$ which is independent of $l,L$, the particular decomposition in \eqref{decomp} and locally uniform in $E$.

\subsection{Bound on the SSF using the Birman-Schwinger principle}
We fix $0<p<1$ and $I\subset \R$ compact. Moreover, 
let $M\in\{1,..., C_l\}$ and $m\in\N$ with $m+1>d/p$ and $m$ odd. In the following we assume without loss of generality that $H_{L,l,M}\ge 0$ for all $L,l>0$. This can always be achieved by adding a constant to $H$.  We denote by $\varrho(\cdot)$ the resolvent set of an operator. Therefore, $1\in \varrho(H_{L,l,M})$ and 
by definition of the spectral shift function, we obtain
\beq
\big| \xi(E,H_{L,l,M},H_{L,l,M-1})\big| = \big| \xi((E+1)^{-m},(H_{L,l,M}+1)^{-m},(H_{L,l,M-1}+1)^{-m})\big|.
\eeq
We define 
\beq\label{eqWM}
W_M:=  (H_{L,l,M-1}+1)^{-m} -  (H_{L,l,M}+1)^{-m}.
\eeq
Then the assumption $m+1>d/p$ implies that $W_M\in \S^p$, see Lemma \ref{lemma1}. 
The Birman-Schwinger principle, see Lemma \ref{app:lem}, implies that 
\begin{align}
\big| \xi\big((E+1)^{-m}&,(H_{L,l,M}+1)^{-m},(H_{L,l,M-1}+1)^{-m}\big)\big|\notag\\
&\leq
\liminf_{\substack{ \eps\searrow 0}}\, \big\| |W_M|^{1/2} \frac 1 {(H_{L,l,M}+1)^{-m} -(E-i \eps+1)^{-m}} |W_M|^{1/2}  \big\|_p^p\notag\\
&\leq
(E+1)^{m p}\liminf_{\substack{ \eps\searrow 0}}\, \big\| |W_M|^{1/2} \frac {(H_{L,l,M}+1)^m}{h_{\wtilde E}(H_{L,l,M})} \frac 1{H_{L,l,M} -\wtilde E} |W_M|^{1/2} \big\|_p^p,\label{pf:eq1}
\end{align}
where $\wtilde E:=E-i\eps$ and 
\beq
h_E(x):=\frac{(1+x)^m - (1+E)^m}{x-E}=\sum_{k=1}^{m-1} (1+x)^k (1+E)^{m-k-1}.
\eeq
 We note that we allow the right hand side of \eqref{pf:eq1} to be $\infty$.  In particular, using Fatou's lemma we obtain for all $E\in\R$  the bound
 \begin{align}
\E\big[ \big| \xi(E,H_{L,l,M},H_{L,l,M-1})\big|\big] & \notag\\
\leq
(E+1)^{m p} \liminf_{\substack{ \eps\searrow 0}}&\ \E\big[
 \big\| |W_M|^{1/2} \frac {(H_{L,l,M}+1)^m}{h_{\wtilde E}(H_{L,l,M})} \frac 1{H_{L,l,M} -\wtilde E} |W_M|^{1/2} \big\|_p^p
\big].\label{pf:eq11}
 \end{align}
In the following we write $H_M:=H_{L,l,M}$ to shorten notation.
\subsection{Inserting higher powers of resolvents}
To bound the expectation value of \eqref{pf:eq1}, we aim at applying the following weak-$L^1$ bound for resolvents due to \cite{artRSO2006AizEtAl2}. This is the main ingredient to the proof. We formulate the lemma in a general way:
\begin{lemma}\label{weakL1}
Let $0<p<1$. Then there exists a constant $C_{p}>0$ such that for all self-adjoint operators $A$, all $\eps>0$, all $E\in\R$ and all operators  $M_1, M_2\in \S^2$ and $U_1,U_2\in\S^\infty$ with $U_1,U_2\ge 0$ the inequalities
\beq\label{eq1lemma}
\int_0^1 \d t\,\int_0^1  \d s\, \rho(t) \rho(s) \| M_1 U_1^{1/2}\frac 1 {A+ t U_1 + s U_2 -E+i \eps }U_2^{1/2}M_2\|_2^p \leq C_{p} \|M_1\|_2^p \|M_2\|_2^p
\eeq
and
\beq\label{eq2lemma}
\int_0^1\d t\,  \rho(t)  \big\|M_1 U_1^{1/2}\frac 1 {A+ t U_1 -E+i \eps }U_1^{1/2}M_2\big\|_2^p \leq C_{p} \|M_1\|_2^p \|M_2\|_2^p
\eeq
hold.
\end{lemma}
\begin{proof}
This lemma follows from \cite[Lemma 3.1 and Proposition 3.2]{artRSO2006AizEtAl2} and the layer-cake representation, see also \cite[App. A.3]{MR2303305}. 
\end{proof}

Later on we apply Lemma \ref{weakL1} with $A=H_{L,l,M}$ and $U_1=u_i$ and $U_2=u_j$ for $i,j\in \Z^d$. Before doing so, we need to rewrite the expression \eqref{pf:eq11} in a suitable way to fit the requirements of Lemma \ref{weakL1}. Especially, we need to raise the $\|\cdot\|_p$-norm  appearing  in \eqref{pf:eq1} for some $0<p<1$ to a Hilbert-Schmidt norm and we have to insert Hilbert-Schmidt operators $M_1,M_2$ which are independent of the "random variables" $s$ and $t$. We do both things using the resolvent equation. This is the main part of this subsection. 

Let $\tau\in\N$ even and to be determined later. We apply the resolvent equation $\tau$-times and obtain 
\begin{align}\label{resEq}
\frac 1 {H_{M} -\wtilde  E} 
=
\sum_{k=1}^\tau \Big( \frac 1{H_{M}+1}\Big)^k (\wtilde E+1) ^{k-1}+(\wtilde E+1)^\tau \Big(\frac 1{H_{M}+1}\Big)^\tau \frac 1 {H_{M}-\wtilde E}.
\end{align}
 Then the triangle-like inequality \eqref{ptriangle} and \eqref{resEq} imply that
\begin{align}
 \big\| |W_M|^{1/2}\frac {(H_{M}+1)^m}{\wtilde h_{E}(H_{M})} \frac 1{H_{M} -\wtilde E}&  |W_M|^{1/2} \big\|_p^p\notag\\
\leq&
\sum_{k=1}^\tau  |\wtilde E+1| ^{p(k-1)} \big\| |W_M|^{1/2} \frac{g_{\wtilde E}(H_{M})}{(H_{M}+1)^{k-1}}  |W_M|^{1/2} \big\|_p^p
\notag\\
&+ |\wtilde  E+1|^{\tau p} \big\| |W_M|^{1/2}\frac{g_{\wtilde E}(H_{M})}{(H_{M}+1)^{\tau}}\frac {1} {H_{M} -\wtilde  E} |W_M|^{1/2}  \big\|_p^p,
\label{pf:eq111}
\end{align}
where $g_{z}(x):=(x+1)^{m-1}/h_{ z}(x)$. 
We note that $g_{z}$ is the rational function
\beq\label{gE}
g_{z}(x) =\frac{(x+1)^{m-1} (x-z)}{(x+1)^m-(z+1)^m}=\prod_{i=1}^{m-1} \frac{x-a_i(z)}{x-b_i(z)}
\eeq
for some $(a_i(z))_{i=1}^{m-1}\subset\C$ and $(b_i(z))_{i=1}^{m-1}\subset\C$ depending continuously on $z\in\C$.
The assumption $m\in\N$ odd implies that $x=E$ is the only real-valued solution of the equation $(x+1)^m=(E+1)^m$, i.e. $\Im b_i(E)\neq 0$ for all $i=1,...,m-1$. Moreover, $b_i(E)\neq b_j(E)$ for $i\neq j$ and $\sup_{x\in\R} |g_E(x)|<\infty$. By continuity of the $b_i(z)$ for $i=1,...,m-1$ this implies
\beq\label{z?}
\sup_{x\in\R,\eps\in(-\eps_0,\eps_0)} |g_{E+i\eps}(x)|<\infty
\eeq
 for some $\eps_0>0$. In the following we restrict ourselves to $\eps<\eps_0$ such that \eqref{z?} holds.

We estimate the expectation of the first part of \eqref{pf:eq111} with the generalized H\"older inequality $\|A B C\|_p \leq \|A\|_{2p} \|B\| \|C\|_{2p}$. This gives for $k=1,...,\tau-1$
\beq\label{pf:eq1111}
 \E \big[\big[ \big\| |W_M|^{1/2} \frac{g_{\wtilde E}(H_{M})}{(H_{M}+1)^{k-1}}  |W_M|^{1/2} \big\|_p^p
 \big] \leq \big\|\frac{g_{\wtilde E}(\,\cdot\,)}{(\,\cdot\, +1)^{k-1}}\big\|_\infty\,  \E \big[\big[\big\| W_M \big\|^p_p \big]<C_1
\eeq
for some $C_1>0$ independent of all relevant quantities where the latter is finite due to \eqref{z?} and  Lemma \ref{lemma1}. 

We note that using the covering condition {\upshape(V2)}, we can write the identity on $L^2(\R^d)$ as $\id= \sum_{n\in\Z^d} u_n v$ for some $v\in L_c^\infty(\R^d)$ with $\|v\|_\infty\leq \frac 1 {C_{u,-}}$.  Hence, using \eqref{ptriangle}, we estimate the second part of \eqref{pf:eq111} by
\begin{align}
&\||W_M|^{1/2}\frac{1}{(H_{M}+1)^{\tau/2}}\frac {1} {H_{M} -\wtilde  E}\frac{g_{\wtilde E}(H_{M})}{(H_{M}+1)^{\tau/2}} |W_M|^{1/2} \big\|_p^p
\notag\\
\leq & 
\sum_{i,j\in\Z^d} \big\| |W_M|^{1/2} \frac 1{(H_{M} +1)^{\tau/2}}u_i v\frac {1} {H_{M} -\wtilde  E} u_j v \frac {g_{\wtilde  E}(H_{M})}{(H_{M} +1)^{\tau/2}}  |W_M|^{1/2} \big\|_p^p.
\label{pf:eq2}
\end{align}
We note here that we have actually inserted an identity on $L^2(\Lambda_{L,l})$ but any identity on $L^2(\R^d)$ is trivially an identity on the smaller space as well and one should rather think of  $u_i$ as $\wtilde u_i:=u_i \id_{\Lambda_{L,l}}$ for $i\in\Z^d$.

\noindent
In the following, we write
\beq\label{xyz}
\frac 1{(H_{M} +1)^{\tau/2}} = \frac 1{(H_{M} +1)^{\tau/2-\sigma}}\frac 1{(H_{M} +1)^{\sigma}}
\eeq
where $\sigma\in\N$ is chosen such that $2\sigma>d/2$ and $\tau\in\N$ even such that $\tau/2-\sigma\ge 0$. To keep notation simple and the argument clear, we restrict ourselves for the moment to $\sigma=1$ and therefore to $d\le 3$. We set $V_{ij}:= \omega_i u_i +\omega_j u_j$. 
Then the resolvent equation implies
\beq\label{aaaaaa}
\frac 1{H_{M} +1}= \frac 1{H_{ij} +1} \Big(\id - V_{ij} \frac 1{H_{M} +1}\Big)=\frac 1{H_{ij} +1} T,
\eeq
where $H_{ij} :=  H_{M}- V_{ij}$ indicating that $H_{ij}$ is independent of the random variables $\omega_i$ and $\omega_j$ and we have set $T:=\id - V_{ij} \frac 1{H_{M} +1}$. 
Inserting \eqref{aaaaaa} in \eqref{pf:eq2}, implies 
\begin{align}
&\big\| |W_M|^{1/2} \frac 1{(H_{M} +1)^{\tau/2}}u_i v\frac {1} {H_{M} -\wtilde  E} u_j v \frac {g_{\wtilde  E}(H_{M})}{(H_{M} +1)^{\tau/2}}  |W_M|^{1/2} \big\|_p^p\notag\\
\leq 
&\big\| |W_M|^{1/2} \frac 1{(H_{M} +1)^{\tau/2-1}}T\frac 1{H_{ij} +1}v u_i \frac {1} {H_{M} -\wtilde  E}u_j v\frac 1{H_{ij} +1}   T^*\frac {g_{\wtilde  E}(H_{M})}{(H_{M} +1)^{\tau/2-1}}  |W_M|^{1/2} \big\|_p^p.\label{proof:eq}
\end{align}
Now, we insert four more identities of the form $\id =  \sum_{n\in\Z^d} \chi_n $ in the above and use \eqref{ptriangle} which gives together with \eqref{pf:eq2} that
\begin{align}
\E\big[\||W_M|^{1/2}\frac{1}{(H_{M}+1)^{\tau/2}}&\frac {1} {H_{M} -\wtilde  E}\frac{g_{\wtilde E}(H_{M})}{(H_{M}+1)^{\tau/2}} |W_M|^{1/2} \big\|_p^p\big]\notag\\
\leq &
\sum_{i,j,n,m,a,b\in\Z^d} \E\big[ \|A^\tau_{n}T_{n,m}  B_{m,i} C_{i,j} B^*_{j,b} T_{b,a}^* (D_a^\tau)^*\|_p^p\big].\label{pf:eq6} 
\end{align}
We have set in above
\begin{align}
A^\tau_n:=|W_M|^{1/2} \frac 1{(H_{M} +1)^{\tau/2-1}} \chi_n^{1/2} \quad,\quad T_{n,m}:= \chi_n^{1/2} \big(\id - V_{ij} \frac 1{H_{M} +1}\big)\chi_m^{1/2}
\end{align}
and
\begin{align}
B_{m,i}:= \chi_m^{1/2} \frac 1{H_{ij}+1} v u_i^{1/2} \quad,\quad C_{i,j}:= u_i^{1/2}\frac 1{H_{M} -\wtilde  E} u_j^{1/2}
\end{align}
and 
\beq
D^\tau_a:=|W_M|^{1/2} \frac {g_{\wtilde E}(H_{M})}{(H_{M} +1)^{\tau/2-1}}  \chi_a^{1/2}.
\eeq
We fix $r>0$, such that for $0<p<1$ and we have
\beq
\frac 1 p = \frac 1 2 + \frac 1 r +\frac 1 r.
\eeq
We apply the generalized H\"older inequality to estimate for all $i,j,n,m,a,b\in\Z^d$
\begin{align}\label{pf:eq3}
\E\big[ \|A^\tau_{n}T_{n,m}  B_{m,i} C_{i,j} &B^*_{j,b} T_{b,a}^* (D_a^\tau)^*\|_p^p\big]\notag\\
&\leq  \E\big[ \|A^\tau_n\|_r^{p}\|T_{n,m}\|^p \|B_{m,i} C_{i,j} B^*_{j,b}\|_{2}^{p}  \|T_{b,a}\|^p\|(D_a^\tau)\|_r^{p} \big].
\end{align}
Since $\supp \rho\subset [0,1]$, assumption \eqref{uk-bounds} implies that $\|V_{ij}\|\leq C_{u,+}$. Using this and Lemma ~\ref{lem:SchattenCombesThomas}, it is straight forward to see that for $n,m\in\Z^d$
\beq\label{pf:eq5abbb}
\|T_{n,m}\| \leq C_1 e^{-c_1|n-m|}
\eeq
 for constants $c_1$ and $C_1>0$ independent of all relevant quantities. 
We choose $\tau\in\N$ also such that $r>d/(\tau/2-1)$. Since $g_{\wtilde E}$ has the form \eqref{gE}, we are in position to apply Lemma \ref{lemma2}, proved later on, to obtain for all $n\in\Z^d$ the estimate
\beq\label{pf:eq5ab}
\|A^\tau_n\|_r< C_r e^{-c_r \dist( n,\Gamma_M)} \qquad\text{and}\qquad   \|D^\tau_n\|_r < C_r e^{-c_r \dist( n,\Gamma_M)},
\eeq
where the constants $C_r >0$ and $c_r>0$ are locally uniform in energy $E$, independent of the random variables and all other relevant quantities.
Hence, \eqref{pf:eq5abbb} and \eqref{pf:eq5ab} imply 
\begin{align}
\sum_{n\in Z^d}\|A_n\|_r^{p} \|T_{n,m}\|^p 
&\leq \sum_{n\in Z^d} C_r C_1 e^{-c_r \dist( n,\Gamma_M)}  e^{-c_1|n-m|}\notag\\
&\leq C_2 e^{-c_2 \dist(m,\Gamma_M)}
\end{align}
for some constants $C_2>$ and $0<c_2<\min\{c_r,c_1\}$ independent of all relevant quantities, especially of the random variables, and locally uniform in  energy. The same bound holds for $\sum_{a\in\Z^d}  \|T_{b,a}\|^p\|(D_a^\tau)\|_r^{p} $.
Inserting the latter in \eqref{pf:eq3}, gives
\beq\label{pf:eq4}
\sum_{n,a\in Z^d}\eqref{pf:eq3} \leq C_3 e^{-c_3 \dist( m, \Gamma_M)} e^{-c_2 \dist( b,\Gamma_M)}\E\big[ \|B_{m,i} C_{i,j} B^*_{j,b}\|_2^p \big].
\eeq
for constants $c_3>0$ and $C_3>0$ locally uniform in energy $E$ and uniform in all other relevant quantities. 

\subsection{Application of Lemma \ref{weakL1}}
For $i,j\in \Z^d$, $i\neq j$ we write $H_M= H_{ij} + \omega_i \wtilde u_i+ \omega_j \wtilde u_j$, $\wtilde u_i:= u_i \id_{\Lambda_{L,l}}$ and $\E^{\perp}_{ij}$ for the expectation with respect to all random variables except $\omega_i$ and $\omega_j$. With this notation at hand the expectation in \eqref{pf:eq4} can be written as
\begin{align}\label{bbbbb}
\E\big[ \|B_{m,i}& C_{i,j} B^*_{j,b}\|_2^p \big]\notag\\
& = \E^{\perp}_{ij}\Big[ \int_0^1\int_0^1 \d t\, \d s\, \rho(t) \rho(s)\big\|B_{m,i} \wtilde u_j^{1/2}\frac 1 {H_{ij}+ t \wtilde u_j + s\wtilde  u_i -E + i\eps}\wtilde u_i^{1/2} B^*_{j,b} \big\|_2^p\Big].
\end{align}
The latter is precisely of the form needed in Lemma \ref{weakL1} where we note that $B_{m,i}$ and $B^*_{j,m}$ are independent of $t$ and $s$. Therefore, Lemma \ref{weakL1} implies that there exists a constant $C_p>0$ independent of all relevant quantities such that
\beq
\eqref{bbbbb}
 \leq C_{p}\, \E^{\perp}_{ij}\big[ \|B_{m,i}\|_2^p \|B^*_{j,b}\|_2^p\big]\label{pf:eq55}.
\eeq
The same holds for $i=j$ along the very same lines using \eqref{eq2lemma} in Lemma \ref{weakL1}. 
Since we restricted ourselves to $d\le 3$, we have $2>d/2$ and Lemma \ref{lem:SchattenCombesThomas} provides for all $n,m\in\Z^d$ the bound
\beq\label{pf:eq5}
\|B_{n,m}\|_2\leq C_4 e^{-c_4 |n-m|}
\eeq
for some constants $c_4>0$ and $C_4>0$ independent of all relevant quantities, especially independent of all random variables. Moreover, we used in the latter that $u_m\leq C_{u,+}\sum_{x\in J} \chi_x $ for some finite index set $J\subset \R^d$ only depending on the support of $u$, see {\upshape(V2)} and we also used $\|v\|_\infty\le \frac 1 {C_{u,-}}$. Finally, the bounds  \eqref{pf:eq3}, \eqref{pf:eq4}, \eqref{bbbbb} and the latter imply that there exist constants $c_5>0$ and $C_5>0$ independent of all relevant quantities and locally uniform in the energy such that
\begin{align}\label{pf:eq7}
\eqref{pf:eq6} \leq
C_5 \sum_{i,j,b,m\in\Z^d}    e^{-c_5 |b-M_0|} e^{-c_5 |m-M_0|} e^{-c_5 | b - i|}e^{-c_5 |j -m|},
\end{align}
where $M_0$ is the centre of the hypercube $\Gamma_M$. Estimating the above sum by the corresponding integral it is straightforward to see that 
\beq
\eqref{pf:eq7}\leq C_6<\infty,
\eeq
where the constant $C_6>0$ is locally uniform in the energy $E$ and independent of all other relevant quantities. This together with the bounds \eqref{pf:eq11}, \eqref{pf:eq111} and \eqref{pf:eq1111} give \eqref{pf:eq8} and the assertion follows in a neighbourhood of $E\in I $ where we recall that we assumed $d\le 3$. Since $I$ is compact the result follows. 

 For $d>3$, the proof proceeds along the very same lines. The only difference is that for $d>3$ we have $\rho>1$ in \eqref{xyz}. In this case we substitute \eqref{aaaaaa} by
\beq\label{horrible}
\frac 1{(H_{M} +1)^\rho}= \sum_{k=1}^m A_{k,ij} T_k
\eeq
valid for some $m\in\N$ and
for some operators $A_{k,ij}\in\S^\infty$, $k=1,...,m$, which are independent of the random variables $\omega_i$ and $\omega_j$ and satisfy $\chi_x A_{k,ij}\chi_y\in \S^2$ for $x,y\in\R^d$, and some bounded operators $T_k$. Moreover, there exists constants $c>0$ and $C>0$ independent of all relevant quantities such that for $x,y\in\R^d$ and $k=1,..,m$ the bounds
\beq\label{horrible2}
\|\chi_x A_{k,ij}\chi_y\|_2\leq C e^{-c|x-y|}\quad\text{and}\quad \|\chi_x T_{k}\chi_y\|\leq C e^{-c|x-y|}
\eeq
hold.
Equation $\eqref{horrible}$ follows from an iteration of the resolvent equation where the operators $A_{k,ij}$ and $T_k$ are expressions given in terms of the resolvents of $H_{ij}$ and $H_M$. Then the bounds \eqref{horrible2} follow from Lemma  \ref{lem:SchattenCombesThomas}. For $d>3$ the assertion follows from using \eqref{horrible} and \eqref{horrible2} in the subsequent steps of the proof.

\section{Proof of Theorem \ref{thmgenW} and Theorem \ref{th:LowerWegner}}\label{proofremaining}

\begin{proof}[Proof of Theorem \ref{thmgenW}]
The proof of Theorem \ref{thmgenW} follows along the same lines as the proof of Theorem \ref{lem:SSFEstimate}. In the case of a perturbation by a potential $W\in L_c^\infty(\R^d)$ we set 
\beq
W_M := (H_L+1)^m - (H_L+W +1)^m
\eeq 
instead of \eqref{eqWM} where $m\in\N$ is chosen such that $W_M\in \S^p$ for some $0<p<1$. Then we  use Lemma \ref{app:lem} and the proof follows along the very same lines as the one of Theorem~\ref{lem:SSFEstimate} where we note that Lemma \ref{lemma1} and Lemma \ref{belwo} also hold in the case of $W_M$ given above with $\dist( x,\Gamma_1)$ substituted by $\dist(x, \supp(W))$. 
\end{proof}

\begin{proof}[Proof of Theorem \ref{th:LowerWegner}]
The proof of Theorem \ref{th:LowerWegner} follows word for word from the proof of \cite[Thm. 2.3]{doi:10.1093/imrn/rnx092}. The only difference is that we use Theorem \ref{lem:SSFEstimate} instead of \cite[Thm. 2.7]{doi:10.1093/imrn/rnx092} to deduce \cite[eq. (3.32)]{doi:10.1093/imrn/rnx092}. \cite[Thm. 2.7]{doi:10.1093/imrn/rnx092} was the only reason for the limitation of the previous result to the region of complete localisation.
\end{proof}

\appendix

\section{Geometric resolvent inequality and Combes-Thomas estimate}	

In this section we consider the following deterministic Schr\"odinger operator
\begin{MyDescription}
\item[(D)] {$H:=-\Delta+V_0+ V$ with two bounded potentials $V_0, V \in L^{\infty}(\R^{d})$ such that $0\le V\le T$ for some finite constant $T>0$ and $-\Delta+V_0\ge 0$.}
\end{MyDescription}

Let $L,l>0$ with $\Lambda_l(x_0)\subset \Lambda_L$ for some $x_0\in\Lambda_L$  and $\Gamma\subset\partial \Lambda_l(x_0)$ open. We denote by $H_{L,l,\Gamma}$ the restriction of $H$ to $\Lambda_{L,l}:=\Lambda_L\setminus \overline{\Lambda_l(x_0)}$ with Dirichlet boundary conditions on $ \partial \Lambda_L$ and Dirichlet boundary conditions on $\Gamma$ and Neumann boundary condition on $\partial \Lambda_l(x_0)\setminus \overline{\Gamma}$, where we also assume that $ \partial\Lambda_l(x_0)\setminus\{\Gamma\cup \partial \Lambda_l(x_0)\setminus{\overline\Gamma}\}$ has Lebesgue measure $0$ on $\partial \Lambda_l(x_0)$, see \cite[Sec. 2]{MR3631314} for the definition of the corresponding Laplace operator. Moreover, we write $|x|$ for the Euclidean norm of a vector $x\in\R^d$. We also recall that $\chi_x:= \id_{\Lambda_1(x)}$ for $x\in\R^d$. 

\begin{lemma}\label{lemma1}
	Assume {\upshape (D)} and $p>0$. Let $m+1>d/p$ and $T>0$.
	There exists a finite constant $C>0$
	 such that for all $L,l>0$,
	$x_{0}\in\Lambda_{L}$ such that $\overline{\Lambda_{l}(x_{0})} \subset \Lambda_{L}$  with $\dist(\Lambda_{l}(x_{0}), \partial  \Lambda_{L})>1$, for all
	 $\Gamma\subset \partial \Lambda_l(x_0)$ open, all hypercubes $\Gamma_1\subset \partial \Lambda_l(x_0)$ of side length $1$ and all
	measurable potentials $V: \Rn \rightarrow [0,T]$ 
 the operator $W_{\Gamma,\Gamma_1}$ satisfies  $W_{\Gamma,\Gamma_1}\in \S^p$ with
\beq
 \|W_{\Gamma,\Gamma_1}\|_p <C<\infty,
\eeq
where $W_{\Gamma,\Gamma_1}:= (H_{L,l,\Gamma}+1)^{-m} -  (H_{L,l,\Gamma\cup\Gamma_1}+1)^{-m}$.
\end{lemma}

We are not aiming at optimality in the above and the assumption $m+1>d/p$ is not optimal. 

\begin{proof}[Proof of Lemma \ref{lemma1}]
We restrict ourselves to $p\leq 1$ in the proof. The proof for $p>1$ follows along the very same lines using the actual triangle inequality where \eqref{ptriangle} is used in the following. 
Let $\Gamma_2:=\Gamma\cup\Gamma_1$. We assume in the following that $\Gamma_2\setminus\Gamma\neq \emptyset$. We write
\begin{align}
W_{\Gamma,\Gamma_1}
&=\sum_{k=0}^{m-1} (H_{L,l,\Gamma}+1)^{-k} \Big(\frac  1 {H_{L,l,\Gamma}+1}- \frac 1  {H_{L,l,\Gamma_2}+1}\Big)(1+H_{L,l,\Gamma_2})^{-(m-1-k)}\notag\\
&= \sum_{k=0}^{m-1} \sum_{x,y\in\Z^d}  (H_{L,l,\Gamma}+1)^{-k} \chi_x\Big(\frac  1 {H_{L,l,\Gamma}+1}- \frac 1  {H_{L,l,\Gamma_2}+1}\Big)\chi_y(H_{L,l,\Gamma_2}+1)^{-(m-1-k)}.
\end{align}
We first show that $(H_{L,l,\Gamma}+1)^{-r} \chi_x \in \S^p$ for $k>d/(2p)$. To see this, we estimate using inequality \eqref{ptriangle}
\begin{align}\label{s1lmeq2}
\|(H_{L,l,\Gamma}+1)^{-k} \chi_x\|_p^p
& \leq \sum_{y_1,...,y_k\in\Z^d} \|\chi_{y_1}(H_{L,l,\Gamma}+1)^{-1}\chi_{y_2}... \chi_{y_k}(H_{L,l,\Gamma}+1)^{-1}\chi_x\|_p^p\notag\\
&\leq \sum_{y_1,...,y_k\in\Z^d} \prod_{i=1}^k  \|\chi_{y_i}(H_{L,l,\Gamma}+1)^{-1}\chi_{y_{i+1}}\|^p_{p k},
\end{align}
where $y_{i+1}= x$. We note that $p k>d/2$. Hence, Lemma \ref{lem:SchattenCombesThomas} implies
\beq\label{Neumann}
\eqref{s1lmeq2}\leq C_1 \sum_{y_1,...,y_k\in\Z^d}  \prod_{i=1}^k e^{-c_1 |y_i-y_{i+1}|}
\eeq
for some $C_1>0$ and $c_1>0$ independent of all relevant quantities. 
Using repeatedly
\begin{equation}
\label{eq:aPriori3}
\sum_{y_2\in\Zd} e^{-c_1|y_{1}-y_{2}|} e^{-c_1|y_{2}-y_{3}|} \leq C_2 e^{-c_1/2 |y_{1}-y_{3}|}
\end{equation}
for some $C_2>0$
in \eqref{Neumann}, we end up with 
\beq\label{eq:lemma111}
\|(H_{L,l,\Gamma}+1)^{-k} \chi_x\|_p <C_{k p}<\infty,
\eeq
where $C_{k p}$ is independent of all relevant quantities. 
Let $k\in\{0,...,m-1\}$. 
Then the H\"older inequality with $1/\infty+ 1/\infty+ 1/p =1/p$ implies
\begin{align}
&\big\| (H_{L,l,\Gamma}+1)^{-k} \chi_x\Big(\frac  1 {H_{L,l,\Gamma}+1}- \frac 1  {H_{L,l,\Gamma_2}+1}\Big)\chi_y(H_{L,l,\Gamma_2}+1)^{-(m-1-k)}\big\|_p\notag\\
\leq&
\big\|\chi_x\Big(\frac  1 {H_{L,l,\Gamma}+1}- \frac 1  {H_{L,l,\Gamma_2}+1}\Big)\chi_y\big\|
\min\big\{\|\chi_y(H_{L,l,\Gamma_2}+1)^{-(m-1-k)}\big\|_p,\|(H_{L,l,\Gamma}+1)^{-k} \chi_x\|_p \big\}.\label{s1lmeq1}
\end{align}
We assume for the moment that $k\ge m-1-k$ and therefore $k\ge (m-1)/2>d/(2p)$.
Taking together \eqref{s1lmeq1} and \eqref{eq:lemma111},  we obtain for $m+1>d/p$ and $k\ge m-1-k$
\begin{align}\label{S1bound}
&\big\| (H_{L,l,\Gamma}+1)^{-k} \chi_x\Big(\frac  1 {H_{L,l,\Gamma}+1}- \frac 1  {H_{L,l,\Gamma_2}+1}\Big)\chi_y(H_{L,l,\Gamma_2}+1)^{-(m-1-k)}\big\|_p\notag\\
\leq&
C_{kp} \big\|\chi_x\Big(\frac  1 {H_{L,l,\Gamma}+1}- \frac 1  {H_{L,l,\Gamma_2}+1}\Big)\chi_y\big\|\leq C_3 e^{-c \dist(x,\Gamma_1)} e^{-c\dist(y,\Gamma_1)}
\end{align}
for some $C_3>0$,
where we used Lemma \ref{belwo} for the last inequality.
The case $k< m-1-k$ follows along the very same lines estimating $\|\chi_y(1+H_{L,l,\Gamma_2})^{-(m-1-k)}\big\|_p$ instead. 

 Since the bound \eqref{S1bound} is summable in $x,y\in\Z^d$, we obtain together with \eqref{ptriangle} that $W_{\Gamma,\Gamma_1}\in \S^p$ with 
\beq
\|W_{\Gamma,\Gamma_1}\|_p^p\leq C_3^p \sum_{x,y\in\Z^d}  e^{-cp\dist(x,\Gamma_1)} e^{-cp\dist(y,\Gamma_1)}= C_4<\infty
\eeq
for some $C_3$ independent of all relevant quantities. 
\end{proof}

\begin{lemma}\label{belwo}
	Assume {\upshape (D)}.
	Let $T>0$. There exist
$C,c>0$ such that for all $x,y\in\mathbb{R}^d$,	 such that for all $L,l>0$,
	$x_{0}\in\Lambda_{L}$ such that $\overline{\Lambda_{l}(x_{0})} \subset \Lambda_{L}$ with $\dist(\Lambda_{l}(x_{0}), \partial  \Lambda_{L})>1$, for all
	 $\Gamma\subset \partial \Lambda_l(x_0)$ open, all hypercubes $\Gamma_1\subset \partial \Lambda_l(x_0)$ of side length $1$ and all
	measurable potentials $V: \Rn \rightarrow [0,T]$ 
\beq
\big\| \chi_x\Big( \frac 1{ H_{L,l,\Gamma}+1} - \frac 1 {H_{L,l,\Gamma\cup\Gamma_1}+1}\Big) \chi_y\big\| \leq C e^{-c\dist ( x,\Gamma_1)} e^{-c\dist(y,\Gamma_1)}. 
\eeq
\end{lemma}

\begin{proof}
We use the notation $\Gamma_2:=\Gamma\cup\Gamma_1$ and assume that $\Gamma_2\setminus\Gamma\neq \emptyset$. 
Let $x\in\R^d$ such that $\dist(x,\Gamma_1)\leq 4$. Then by Lemma \ref{lem:SchattenCombesThomas}
\begin{align}
\big\| \chi_x\Big( \frac 1{H_{L,l,\Gamma}+1} - \frac 1 {H_{L,l,\Gamma_2}+1}\Big) \chi_y\big\|
 \leq &
 \| \chi_x \frac 1 {H_{L,l,\Gamma_2}+1} \chi_y\|+  \| \chi_x \frac 1 {H_{L,l,\Gamma}+1} \chi_y\|\notag\\
\leq &
C_1 e^{-c_1 |x-y|} \leq C_1 e^{-c_2\dist(\Gamma_1,y)}
\end{align}
for some appropriate constants $c_1>c_2>0$ and $C_1>0$ independent of all relevant quantities. 
The same also holds for $y\in\R^d$ such that $\dist(y,\Gamma_1)\leq 4$. 

Hence, we consider in the following $x,y\in\R^d$ such that $\dist(x,\Gamma_1), \dist(y,\Gamma_1)\ge 4$. 
 We introduce a switch function $\psi \in C^{2}(\Lambda_{L,l})$ with $\dist\big(\supp(\psi),\Gamma_1\big)\geq 1/4$,
\begin{equation}
 	\supp(\nabla \psi)\subseteq \Big\{z\in \Lambda_{L,l}: 1/4 \leq \dist\big(z,\Gamma_1\big)\leq 1/2\Big\}
	=:\Omega,
\end{equation}
$ \|\nabla \psi\|_{\infty} \leq 8$ and $1 \geq \psi\geq \id_{\Lambda_{L,l}\setminus \partial \Gamma_{1,-}}$, where $ \partial \Gamma_{1,-}:=\{x\in \Lambda_{L,l}:\ \dist(x, \partial \Gamma_{1})<2\}$. We have
\begin{align}\label{pf:lm2:eq2}
\big\| \chi_x\Big( \frac 1{ H_{L,l,\Gamma}+1} - \frac 1 {H_{L,l,\Gamma_2}+1}\Big) \chi_y\big\|
&=
\big\|\chi_x \big(\frac 1{ H_{L,l,\Gamma}+1} \big[-\Delta,\psi\big]  \frac 1 {H_{L,l,\Gamma_2}+1} \big)\chi_y\big\|
\end{align}
for all $x,y\in \partial  \Gamma_{1,-} ^c$.
Here, the operator $\psi H_{L,l,\Gamma} - H_{L,l,\Gamma_2} \psi = -[-\Delta,\psi]$
is a differential operator of order one acting only on $\supp(\nabla \psi)$. 
Next we cover $\supp(\nabla \psi)$ by the cube $\Lambda_2(k)$ where $k$ is the center of the hypercube $\Gamma_1$ and estimate
\begin{align}\label{pf:lm2:eq42}
\eqref{pf:lm2:eq2}
&\leq  \big\|\chi_x \frac 1{H_{L,l,\Gamma}+1} \id_{\Lambda_2(k)}\big[-\Delta,\psi\big]  \frac 1 {H_{L,l,\Gamma_2}+1} \chi_y\big\|\notag\\
&\leq  \big\|\chi_x \frac 1{ H_{L,l,\Gamma}+1}\id_{\Lambda_2(k)} \big\|\big\|\id_{\Lambda_2(k)}\big[-\Delta,\psi\big]  \frac 1 {H_{L,l,\Gamma_2}+1} \chi_y\big\|.
\end{align}
We note that $\dist(\Lambda_1(y),\Lambda_2(k))\ge 1$ for all $y\in \R^d$ with $\dist(y,\Gamma_1)\ge 4$. Standard arguments provided for example in
     \cite[Lemma 2.5.3]{MR1935594} and the proof of \cite[Lemma 2.5.2]{MR1935594}  imply
\beq\label{pf:lm2:eq33}
\big\|\id_{\Lambda_2(k)}\big[-\Delta,\psi\big]  \frac 1 {H_{L,l,\Gamma_2}+1} \chi_y\big\|\leq C_2  \big\|\1_{\Lambda_3(k)} \frac 1 {H_{L,l,\Gamma_2}+1} \chi_y\big\|,
\eeq
where $C_2>0$ is independent of the precise boundary condition on $\Gamma_1$ and all other relevant quantities. Lemma \ref{lem:SchattenCombesThomas} below implies that
\beq\label{pf:lm2:eq333}
\eqref{pf:lm2:eq33} \leq C_3 e^{-c_3|k-y|}. 
\eeq
for some constants $c_3>0$ and $C_3>0$ independent of all relevant quantities. For the second term in \eqref{pf:lm2:eq42}  Lemma \ref{lem:SchattenCombesThomas} yields
 \beq\label{pf:lm2:eq43}
 \displaystyle\big\|\chi_x\frac 1{ H_{L,l,\Gamma}+1} \id_{\Lambda_2(k)} \big\|\leq C_3 e^{-c_3|x-k|}
 \eeq
  with some appropriate constants $c_3>0$ and $C_3>0$ uniform in all relevant quantities. Putting it all together,  \eqref{pf:lm2:eq42},  \eqref{pf:lm2:eq333} and \eqref{pf:lm2:eq43} imply that 
\beq\label{pf:lm2:eq7}
\eqref{pf:lm2:eq2}\leq  C_4  e^{-c_4|x-k|}e^{-c_4|k-y|} \leq C_5 e^{-c_4 \dist(x,\Gamma_1)} e^{-c_4 \dist(y,\Gamma_1)} 
\eeq
for some appropriate constants $c_4>0$ and $C_4, C_5>0$ which are uniform in all relevant quantities. 
\end{proof}

\begin{lemma}
	\label{lem:SchattenCombesThomas}
	Assume {\upshape (D)}.
	Let $p>d/2$, $z \in \C\backslash [0,\infty)$ and $T>0$. Then, there exist finite constants
	$C_{p,z},c_{p,z}>0$ such that for all $x,y\in\mathbb{R}^d$, such that for all $L,l>0$,
	$x_{0}\in\Lambda_{L}$ such that $\overline{\Lambda_{l}(x_{0})} \subset \Lambda_{L}$, for all
	 $\Gamma\subset \partial \Lambda_l(x_0)$ open and all 
	measurable potentials $V: \Rn \rightarrow [0,T]$ we have the estimate
	\begin{equation}
	\label{eq:SchattenCTStat}
	\|\chi_x (H_{L,l,\Gamma}-z)^{-1}\chi_y \|_p \leq C_{p,z} \, e^{-c_{p,z}|x-y|},
	\end{equation}	
	where the constants $c_{p,z}>0$ and $C_{p,z}>0$ depend continuously on the distance $\dist(z, [0,\infty))$. 
\end{lemma}

\begin{proof}
The proof follows from \cite[Lemma A.2]{doi:10.1093/imrn/rnx092} where we note that the operator norm Combes-Thomas estimate used in the proof there extends to complex energies and mixed boundary conditions as well, see in particular \cite[Thm.\ 2.4.1 and Rmk.\ 2.4.3]{MR1935594}, \cite[Thm.~1]{MR1937430}
and \cite{MR3255146}. 
\end{proof}
To state the next lemma let 
\beq\label{defg}
g(x):=\prod_{i=1}^n \frac{x - a_i}{x - b_i},\quad x\in\R,
\eeq
for some $n\in\N_0$
be a rational function with $(a_i)_{i=1}^n\subset\C$ and $(b_i)_{i=1}^n\subset \C$ such that $b_i\neq b_j$ for $i\neq j$ and $\Im b_i\neq 0$. In the case $n=0$ we set $g\equiv 1$. We define for $x\in\R^d$, $\Gamma,\Gamma_1\subset\partial\Lambda_l$ open with $\Gamma_1$ being a hypercube of side length $1$
\beq
 D^{g,\tau}_{x,\Gamma,\Gamma_1}:=|W_{\Gamma,\Gamma_1}|^{1/2} \frac {g(H_{L,l,\Gamma_2})}{(H_{L,l,\Gamma_2} +1)^\tau}  \chi_x,
\eeq
where we have set $\Gamma_2:=\Gamma\cup\Gamma_1$ and $\tau\in\N$. 

\begin{lemma}\label{lemma2}
	Assume {\upshape (D)} and $p>0$.
	Let $\tau\in\N$ and $p>d/\tau$ and $g$ be of the form \eqref{defg}. Then, there exist finite constants
	$C_{p}, c_{p}>0$ such that for all $x\in\mathbb{R}^d$, for all $L,l>0$,
	$x_{0}\in\Lambda_{L}$ such that $\overline{\Lambda_{l}(x_{0})} \subset \Lambda_{L}$  with $\dist(\Lambda_{l}(x_{0}), \partial  \Lambda_{L})>1$, for all
	 $\Gamma\subset \partial \Lambda_l(x_0)$ open and all hypercubes $\Gamma_1\subset \partial \Lambda_l(x_0)$ of side length $1$ and all
	measurable potentials $V: \Rn \rightarrow [0,T]$
\beq\label{pf:eq5a}
 \| D^{g,\tau}_{x,\Gamma,\Gamma_1}\|_p < C_p e^{-c_p \dist( x,\Gamma_1)},
\eeq
where the constants $c_p$ and $C_p$ depend continuously on $\dist(b_i,\C\backslash [0,\infty))$ for $i=1,...,n$. 
\end{lemma}

\begin{proof}
We restrict ourselves to $p\leq 1$ in the proof. The proof for $p>1$ follows along the very same lines using the actual triangle inequality where \eqref{ptriangle} is used in the following. We write  $D_{x}:=D^{g,\tau}_{x,\Gamma,\Gamma_1}$.
We first focus on the function $g$. 
Since $b_i\neq b_j$, a partial fraction decomposition allows for writing 
\beq\label{lm2:eq111}
g(x)=\prod_{i=1}^{n} \frac{x-a_i}{x-b_i} = 1 - \sum_{i=1}^{n} \frac{d_i}{x-b_i}
\eeq
for some $(d_i)_{i=1}^{n}\subset \C$. 
This and \eqref{ptriangle}  imply that
\begin{align}\label{pf:lm2:main}
\|D_x\|_p^p\leq 
&
 \big\||W_{\Gamma,\Gamma_1}|^{1/2} \frac 1 {(H_{L,l,\Gamma_2} +1)^{\tau}} \chi_x\big\|_p^p\notag\\
&+ \sup_{i=1,...n}|d_i|
 \sum_{i=1}^{n}\big\| |W_{\Gamma,\Gamma_1}|^{1/2}\frac 1 {(H_{L,l,\Gamma_2} +1)^{\tau}}\frac 1 {H_{L,l,\Gamma_2}-b_i} \chi_x\big\|_p^p.
\end{align}
We treat the term $ \big\||W_{\Gamma,\Gamma_1}|^{1/2} \frac 1 {(H_{L,l,\Gamma_2} +1)^{\tau}} \chi_x\big\|_p$ frist. We insert identities $\id =\sum_{k\in\Z^d} \chi_k$ and obtain using the triangle inequality and subsequently the H\"older inequality that 
\begin{align}\label{abc5}
\big\||W_{\Gamma,\Gamma_1}|^{1/2} &\frac 1 {(H_{L,l,\Gamma_2} +1)^{\tau}} \chi_x\big\|^p_p\notag\\
\leq &
\sum_{k_1,...,k_{\tau}\in\Z^d} \big\||W_{\Gamma,\Gamma_1}|^{1/2}\chi_{k_1}\big\|^p_{p (\tau+1)} \prod_{i=1}^{\tau} 
\|\chi_{k_{i}} \frac 1 {H_{L,l,\Gamma_2} +1} \chi_{k_{i+1}}\|^p_{p (\tau+1)},
\end{align}
where $k_{\tau+1} = x$. 
Now, for $k\in\Z^d$ we have
\begin{align}\label{abc2}
\big\||W_{\Gamma,\Gamma_1}|^{1/2}\chi_{k}\|_{p \tau} = \big\|\chi_{k} |W_{\Gamma,\Gamma_1}|\chi_{k} |W_{\Gamma,\Gamma_1}|\chi_{k}\big\|_{p\tau/4}^{1/4}
\leq
 \big\|\chi_{k} W_{\Gamma,\Gamma_1}^2\chi_{k}\big\|_{p\tau/4}^{1/4},
\end{align}
where we used the operator inequality $|W_{\Gamma,\Gamma_1}|\chi_{k} |W_{\Gamma,\Gamma_1}| \leq |W_{\Gamma,\Gamma_1}|^2\leq  W_{\Gamma,\Gamma_1}^2$ in the latter. We insert another identity and estimate
\begin{align}\label{abc}
 \eqref{abc2}
 \leq \sum_{j\in\Z^d}  \big\|\chi_{k} W_{\Gamma,\Gamma_1}\chi_jW_{\Gamma,\Gamma_1} \chi_{k}\big\|_{p\tau/4}^{1/4}
 \leq \sum_{j\in\Z^d}  \big\|\chi_{k} W_{\Gamma,\Gamma_1}\chi_j\big\|_{p\tau/2}^{1/4}\big\|\chi_j W_{\Gamma,\Gamma_1} \chi_{k}\big\|_{p\tau/2}^{1/4},
\end{align}
where we implicitly assumed here that $p\tau/4 \ge 1$. Otherwise, one has to use \eqref{ptriangle} and adapt the exponents accordingly. 
Since we assumed $p \tau>d$, we obtain that $p\tau/2>d/2$. Hence, we can apply Lemma \ref{belwo} and obtain
\beq\label{abc4}
\eqref{abc}\leq  C_1 \sum_{j\in\Z^d}  e^{-c_1\dist(k,\Gamma_1)} e^{-c_1\dist(\Gamma_1,j)} e^{-c_1\dist(j,\Gamma_1)} e^{-c_1\dist(\Gamma_1,k)}\leq C_2 e^{-c_2\dist(\Gamma_1,k)}
\eeq
for some constants $c_1,c_2,C_1,C_2>0$ independent of all relevant quantities. Moreover, since $p\tau>d/2$, Lemma \ref{lem:SchattenCombesThomas} implies for $j,k\in\Z^d$ that 
\beq\label{abc3}
\big\|\chi_k \frac 1 {H_{L,l,\Gamma_2} +1} \chi_j\|_{p \tau}\leq C_3 e^{-c_3 |j-k|}.
\eeq
for constants $c_3,C_3>0$ independent of all relevant quantities. 
Inserting \eqref{abc4} and \eqref{abc3} in \eqref{abc5} implies 
\beq
\eqref{abc5} \leq C_4 e^{-c_4 \dist(x,\Gamma_1)}
\eeq
for constants $c_4,C_4>0$ independent of all relevant quantities. 

The remaining term in \eqref{pf:lm2:main} is estimated along the very same lines. This gives the assertion. 
\end{proof}

\section{Bound on the spectral shift function and the Birman-Schwinger principle}


\begin{lemma}\label{app:lem}
Let $A,B$ be self-adjoint operators in a Hilbert space $\H$ with purely discrete spectra and we assume $W:=B-A\in\S^p$ for some $0<p\leq 1$. Then for all $E\in\R$ 
\beq\label{app:lm}
\big|\xi(E,A,B)\big| \leq \liminf_{\substack{z\to E\\ \Im z\neq 0}}\, \big\| |W|^{1/2} \frac 1 {A-z} |W|^{1/2} \big\|_p^p,
\eeq
where we allow the right hand side to be $\infty$. 
\end{lemma}

The latter lemma is well-known for all energies $E\in\R$ such that the limit $|W|^{1/2} \frac 1 {A-E-i \eps} |W|^{1/2} $  exists in $\S^p$ as $\eps\to 0$, see  \cite{MR1208792}. This follows from the Birman-Schwinger principle. Since we want the result for all energies allowing the right hand side to be $\infty$ as well, we give a proof for completeness. 

\begin{proof}
Let $E\in \R$ and $P:=P_{\{E\}}(A)$ be the orthogonal projection on the spectral subspace of $A$ corresponding to $E\in\R$. Obviously, $P=0$ if $E\in\rho(A)$.  
First we restrict ourselves to the case $W P= 0$ and we write $\wtilde A:= P^\perp A P^\perp$, where $P^\perp:=\id-P$. 
For $\lambda\in(0,1]$ we obtain 
\beq
\frac 1 {B_\lambda-E} = \frac 1 {\wtilde A-E} \Big(\id -  (B_\lambda- \wtilde A)\frac 1 {\wtilde A-E}\Big)^{-1},
\eeq
where $B_\lambda:= A+ \lambda V$. $W P=0$ implies $B_\lambda -\wtilde A= \lambda W + P A P^\perp +P^\perp A P=\lambda W$.
Now, $E$ is an eigenvalue of $B_\lambda$ if and only if $1$ is an eigenvalue of $\lambda W\frac 1 {\wtilde A-E}$. Since the spectrum of operators $R$ and $S$ satisfy $\sigma(RS)\setminus\{0\} = \sigma(SR)\setminus\{0\}$, the latter is equivalent to $\sign(W) |W|^{1/2} \frac 1 {\wtilde A-E} |W|^{1/2}$ having an eigenvalue $1/\lambda$. Considering the eigenvalues of $B_\lambda$ as functions of $\lambda$, we see that 
\begin{align}
\big|\xi(E,A,B)\big| \leq \#\big\{\lambda:\ E \in \sigma(B_\lambda)\big\} 
&=\# \big\{1/\lambda\in \sigma\big(\sign(W) |W|^{1/2} \frac 1 {\wtilde A-E} |W|^{1/2}\big) \big\}\notag\\
&\leq \big\||W|^{1/2} \frac 1 {\wtilde A -E} |W|^{1/2}\big\|_p^p
\label{12345}
\end{align}
 where  $p>0$ and we used $\|\sign(W)\|=1$ in the last inequality. Now
\beq
\eqref{12345} = \liminf_{\substack{z\to E\\ \Im z\neq 0}}\, \big\||W|^{1/2} \frac 1 {\wtilde A-z} |W|^{1/2}\big\|_p^p
=\liminf_{\substack{z\to E\\ \Im z\neq 0}}\, \big\||W|^{1/2} \frac 1 {A-z} |W|^{1/2}\big\|_p^p,
\eeq
where the first equality follows from $W\in \S^p $ and the strong convergence $\frac 1 {\wtilde A-z}\to\frac 1 {\wtilde A-E}$ as $z\to E$ and the second equality follows from $W P= 0$. 

In the case $W P\neq 0$ let $\varphi$ be such that $WP\varphi \neq 0$, $P\varphi =\varphi$ and $\|\varphi\|_2=1$. Then we have for $0<p\leq 1$ the lower bound
\begin{align}
\big\||W|^{1/2} \frac 1 {A-z} |W|^{1/2}\big\|_p^p &\geq \big\||W|^{1/2} \frac 1 {A-z} |W|^{1/2}\big\|^p_1\notag\\
& \geq \big| \Im \Tr \big( |W|^{1/2} \frac 1 {A-z} |W|^{1/2}\big)\big|^p.\label{appeq11}
\end{align}
Let $Q=|\varphi\>\<\varphi|$, $(a_n)_{n\in\N}$ and $(\varphi_n)_{n\in\N}$ be the eigenvalues and eigenvectors of $A$. We chose the eigenvectors such that $\varphi_1 =\varphi$. Then we obtain
\begin{align}
\eqref{appeq11}&\geq \big|  \Tr \big(Q |W|^{1/2}\Im \frac 1 {A-z} |W|^{1/2}Q\big)\big|^p\notag\\
&=\Big|\sum_{n\in\N} \Im  \frac 1 {a_n-z} \big|\<|W|^{1/2}\varphi,|W|^{1/2}\varphi_n\>\big|^2\Big|^p \label{appeq1}.
\end{align}
Since $\Im z\neq0 $ and $\Im  \frac 1 {a_n-z} \geq 0$ has the same sign for all $n\in\N$ we obtain 
\beq
\eqref{appeq1} \geq \frac 1 \eps \<|W|^{1/2}\varphi,  |W|^{1/2}\varphi\>
\eeq
and therefore $\liminf_{\substack{z\to E\\ \Im z\neq 0}} \big\||W|^{1/2} \frac 1 {A-z} |W|^{1/2}\big\|_p^p =\infty$. This trivially gives the bound \eqref{app:lm} in the case $W P\neq 0$.
\end{proof}

\section*{Acknowledgements}
M.G. thanks Sasha Sodin and Adrian Dietlein for interesting discussions on the topic and suggestions on an earlier version of the paper.

\newcommand{\etalchar}[1]{$^{#1}$}


\end{document}